%% file: arXiv-version.tex
\documentclass[english]{lipics-v2021}

\usepackage[T1]{fontenc}
\usepackage[utf8]{inputenc}
\usepackage{microtype}
\usepackage{amsmath,amssymb,amsthm}
\usepackage{mathtools}
\usepackage{todonotes}
\usepackage[normalem]{ulem}
\input{config}

\title{Polyhedral Methods for Cooperative Games: Small Lifts and Hard Faces}

\titlerunning{Polyhedral Methods for Cooperative Games}

\author{Michel Grabisch}{Charles University, Faculty of Mathematics and Physics, Department of Applied Mathematics, Czechia \\ Universit\'e Paris 1 Panth\'eon-Sorbonne, Centre d'Economie de la Sorbonne, France \and \url{https://sites.google.com/site/michelgrabisch/}} {michel.grabisch@univ-paris1.fr}{https://orcid.org/0000-0002-3283-1496}{Partially supported by the AGATE project funded  from the Horizon Europe Programme under Grant Agreement No. 101183743.}

\author{Hans Raj {Tiwary}}{Charles University, Faculty of Mathematics and Physics, Department of Applied Mathematics, Czechia \and \url{https://kam.mff.cuni.cz/~hansraj/}}{hansraj@kam.mff.cuni.cz}{https://orcid.org/0000-0003-1903-1600}{Partially supported by the AGATE project funded  from the Horizon Europe Programme under Grant Agreement No. 101183743.}
\authorrunning{M. Grabisch, H. R. Tiwary}

\Copyright{Michel Grabisch, Hans Raj Tiwary}

\EventShortTitle{arXiv}
\category{}
\relatedversion{}
\supplement{}

\nolinenumbers

\theoremstyle{plain}
\theoremstyle{definition}

\begin{document}

\maketitle
\begin{abstract}
We study the computational complexity of fundamental algorithmic problems — membership testing, separation, valid-inequality testing, and linear optimization — over polytopes and cones arising from cooperative games (also known as pseudo-Boolean functions).  A central obstacle in the study of such problems is that a general cooperative game on
$n$ players requires $2^n-1$ values to be described, so the input size is $2^n-1$ for a game with $n$ players, making these computational tasks theoretically trivial. Restricting to $k$-additive games reduces
the input size to $\Theta(n^k)$, making such games a natural target
for meaningful questions about the existence of efficient algorithms.

On the positive side, for the cone of $k$-additive $(k{-}1)$-monotone games, we give a complete characterization of its extreme rays and derive an
$O(n^k)$ bound on extension complexity, yielding a geometry-based proof and generalization of a result of Billionnet and Minoux. We also give an explicit extended formulation of size $O(n^k)$ for the core of $k$-additive $l$-monotone games for $l\in\{k-1,k\}$, allowing all four problems to be solved by an explicit, compact LP representation, rather than reliance on an optimization or separation oracle. 

On the negative side, we show that for $l \leq k-2$ the cone of
$k$-additive $l$-monotone games is computationally intractable: membership testing is not in NP (unless NP\,=\,coNP), valid-inequality testing is NP-complete, and extension complexity is at least $1.5^{n-l}$. Our hardness reduction works by identifying, through a sequence of facial operations on the dual cone, a copy of the correlation polytope — a canonical hard $0/1$ polytope for which the same problems are known to be intractable. Our hardness results yield, as a special case, a result of Crama and of Gallo and Simeone. Furthermore, our hardness results also explain the lack of any good characterization of the extreme rays of the cone of $k$-additive $(k{-}2)$-monotone games.  

Together, for fixed \(k\), the monotonicity cones exhibit a sharp transition between polynomial-size extended formulations at orders \(l\ge k-1\) and exponential extension complexity at orders \(2\le l\le k-2\).
\end{abstract}

\section{Introduction}
Cooperative games with transferable utility -- also known as pseudo-Boolean functions -- assign a real value $v(S)$ to
every subset $S$ of players. This is useful in settings where one wishes to capture, for example, the benefit of collective action.
In the study of cooperative games, some central algorithmic tasks are to reason about a set of so called stable solutions: the \emph{core} $C(N,v)$ of a game $v$ with set of players $N$, and about testing structural constraints such as supermodularity and its higher-order generalizations. One can phrase these questions in polyhedral terms. The core is by definition a polytope and testing supermodularity, or more generally $l$-monotonicity, is equivalent to membership testing in the set of all such games, which is known to be a polyhedral cone. 

These polyhedra are described by exponentially many inequalities in general. However, these inequalities are known, and since a general cooperative game is described by $2^n-1$ arbitrary values, the computational tasks at hand can be performed ``efficiently''. 
However, instead of the values for individual sets, one may represent a game also by its so-called M\"obius coefficients.
A widely used restriction on cooperative games, $k$-additivity,
limits the Möbius coefficients of subsets to size at most $k$, the rest being zero, cutting the description to $O(n^k)$ parameters.

For fixed $k$, the question becomes: which computational problems over the associated core and monotonicity cones admit polynomial-time algorithms and which are NP-hard?  Since the input size is no longer exponential in the number of players, these questions become nontrivial.

Previous work has addressed membership in the core \cite{DP94}, or checking supermodularity for $k$-additive games for $k\leqslant 4$ \cite{BM85, Crama1989, GS89}, but the results are limited in scope.

\subsection{Our approach}

We study the problems of checking $l$-monotonicity, or membership testing and optimization over the core, through the lens of polyhedral computation, and motivate our work by two questions: 
\begin{enumerate}
\item When do the core and $l$-monotonicity cones admit polynomial-size extended formulations, and when they do, can one be written down explicitly?
\item When they do not, can one prove hardness, for example, by embedding a canonical hard polytope into the facial structure?
\end{enumerate}



We prove the following results:
\begin{enumerate}
\item \textbf{Cone of $k$-additive $(k{-}1)$-monotone games: $\SM{k}{k-1}(n)$.}
      We give a complete description of the extreme rays of $\SM{k}{k-1}(n)$
      and deduce $\xc{\SM{k}{k-1}(n)} = O(n^k)$.  This implies that one can check whether a $k$-additive game is $(k{-}1)$-monotone by solving a linear program whose size is linear in the input size. As a corollary, we obtain a geometry-based proof and generalization of a result of Billionnet and Minoux for $k=3$ \cite{BM85}.

\item \textbf{Cone of $k$-additive $l$-monotone games, $l \leq k-2$: $\SM{k}{l}(n)$.}
      We show that for every pair $(k,l)$ with $2 \leq l \leq k-2$:
      \begin{itemize}
        \item Membership testing in $\SM{k}{l}(n)$ is not in NP unless NP\,=\,coNP;
        \item Valid inequality testing over $\SM{k}{l}(n)$ is NP-complete;
        \item Separation over $\SM{k}{l}(n)$ is NP-hard;
        \item The extension complexity of ${\SM{k}{l}(n)}$ is at least $1.5^{n-l}$.
      \end{itemize}
       These results show that checking whether a $k$-additive game is $l$-monotone is hard for every $l\leqslant k-2$ for every $k\geqslant 4$. This is a strengthening and generalization of a result of Crama and of Gallo and Simeone \cite{Crama1989, GS89} for $k=4, l=2$. Furthermore, these results explain why no efficiently checkable characterization of the extreme rays of these cones is known since the existence of such a characterization would imply NP=coNP.
\item \textbf{Core of $k$-additive $l$-monotone games, $k-1\leqslant l\leqslant k$: $C(N,v)$.}
      For $l=k$, we give an explicit polytope $Q_k(N,v)$ of size $O(n^k)$ that projects onto the core $C(N,v)$. For $l=k-1$, we obtain an explicit extended formulation of size $O(n^k)$ by writing the core as a Minkowski sum of simple polytopes.  Consequently, membership testing, valid inequality testing, separation, and linear optimization over the core can all be solved by a single linear program of size linear in the input size. Earlier results required separate approaches for each problem  or ellipsoid/oracle based transformations \cite{DP94, Edmonds1970Submodular}.

\end{enumerate}

Taken together, results (1)--(2) establish a sharp dichotomy: the
computational problems are tractable (in a strong, LP-based sense) over the monotonicity cones when the monotonicity order is at least $k-1$, and become NP-hard as soon as the order drops to $k-2$ or below. 


\subsection{Related work}

The study of algorithmic problems over the core and of supermodular games has a
long history.  Because supermodular games are equivalent to submodular
functions via negation, linear optimization over the core of $v$ corresponds to linear optimization over the base polytope of the extended polymatroid of $-v$ \cite{schrijver2003combinatorial, fujishige05}. Thus the greedy algorithm solves linear optimization over
the core in polynomial time \cite{Edmonds1970Submodular}. For $k$-additive $k$-monotone games with $k \geq 2$, Deng and Papadimitriou gave a polynomial-time membership-testing algorithm via a max-flow reduction \cite{DP94}.  

A polynomial-time algorithm for testing whether a $3$-additive game is $2$-monotone, was given by Billionnet and Minoux in \cite{BM85}. However, the result uses specific features of $3$-additive $2$-monotone games that do not generalize to $k$-additive $(k{-1})$-monotone games. The hardness direction was initiated
by the NP-hardness results for checking $2$-monotonicity, also called supermodularity, of $4$-additive games by Crama \cite{Crama1989} who proved weak NP-hardness and by Gallo and Simeone \cite{GS89} who proved strong NP-hardness. Again the results do not hold for testing $(k{-2})$-monotonicity of $k$-additive games. Our result (2) gives a uniform hardness result for all $l \leq
k-2$ and all $k\geqslant 4$.

Extension complexity has been used extensively in Combinatorial Optimization \cite{Kaibel2011ExtendedFormulations, Wolsey2011UsingExtendedFormulations} but appears to be new as a tool for studying cooperative game polytopes.  The correlation polytope and its connection to the cut polytope, along with the hardness results we use, are due to Pitowsky \cite{Pitowsky91}, De Simone \cite{DESIMONE199071}, and Karp \cite{Karp72}; exponential lower bounds on its extension complexity appear in \cite{FioriniMPTW15, KaibelW15}.

If one can give a superpolynomial lower bound on the extension complexity, then one can rule out a large class of LP based methods. A superpolynomial lower bound on extension complexity does not rule out polynomial time algorithm for linear optimization over the polytope \cite{Rothvoss17}, but often -- and this is indeed the case in our work -- one can identify a computationally hard polytope embedded in the polytope at hand. Such ``facial embeddings'' have been used to prove extension complexity lower bounds for many polytopes relying on the bounds for the correlation polytope \cite{AvisT15, FioriniMPTW15}. In these results, known NP-hardness proofs were exploited to embed the correlation polytope into a target polytope. In this paper, we do the opposite: namely, we employ such a facial embedding to derive computational hardness.

\section{Preliminaries}
\subsection{Polyhedra and Computation}\label{subsec:problems}
For polyhedron $P\subseteq\BR^d$ we will be interested in four natural computational problems:

\begin{problem}[SMEM]\label{probmem}
    Given $y\in\BR^d$ is $y\in P$?
\end{problem}

\begin{problem}[SVAL]\label{probval}
    Given $\alpha\in\BR^d, \beta\in\BR,$ is $\alpha^\intercal x\leqslant \beta$ valid for $P$?    
\end{problem}

\begin{problem}[SSEP]\label{probsep}
    Given $y\in\BR^d$ output YES if $y\in P$, otherwise give an inequality $\alpha^\intercal x\leqslant \beta$ valid for $P$ such that $\alpha^\intercal y>\beta.$
\end{problem}

\begin{problem}[SOPT]\label{probopt}
    Given $c\in\BR^d$ find $x\in P$ that maximizes $c^\intercal x$. Output UNBOUNDED if no such point exists.
\end{problem}

For well-described polyhedra, that is, polyhedra for which an upper bound on the number of bits required to define any facet and any vertex are known, together with a point inside the polyhedra and some inner and outer balls are known, the four problems: SMEM, SVAL, SSEP, and SOPT, are polynomial-time equivalent to each other \cite{GLS1988}. However, the proof of equivalence uses the ellipsoid algorithm and so having an extended formulation of polynomial size is beneficial because one can avoid the use of the ellipsoid algorithm. All the polyhedra considered by us are well-described so we will omit any mention of these extra assumptions in the following and assume equivalence of these problems.

While a polynomial time algorithm for one of the above problems gives polynomial time algorithm for all, slight care should be taken in translating hardness results for one into another. For example, for a polyhedron whose inequalities are known SMEM is in coNP while SVAL is in NP whereas for a polyhedron whose vertices and extreme rays are known, SMEM is in NP while SVAL is in coNP. So a typical (co)NP-completeness result needs to be appropriately modified. SSEP and SOPT are not decision problems so only NP-hardness can be shown for them.


\subsubsection{Polar Duality}\label{subsec:duality}
    Let $P\subseteq\BR^d$. The \emph{polar dual} of $P$ -- denoted by $P^\polar$ -- is defined to be the set $\{y~|~y^\intercal x \leqslant 1, ~\forall x\in P\}.$

If $P$ is closed and convex, and contains the origin then $(P^\polar)^\polar=P.$

Let $P$ be a polyhedron containing the origin. That is, suppose $P=\conv{V}+\cone{W}=\{x~|~ Ax\leqslant 1, Bx\leqslant 0\}$ then, $P^\polar=\conv{A \cup \{0\}}+\cone{B}=\{x~|~ Vx\leqslant 1, Wx\leqslant 0\}.$ Here $\conv{A}$ for a matrix $A$ represents the convex hull of the set of rows of $A$ treated as vectors. Similarly  $\cone{B}$ is the conic hull of the set of row vectors of $B$.

\begin{theorem}\label{thm:prob_equiv}
    Let $P$ be a well-described polyhedron containing the origin. The problems SMEM, SVAL, SSEP, SOPT are equivalent to the problems SVAL, SMEM, SOPT, SSEP, respectively, over the polar dual $P^\polar$.    
\end{theorem}
\begin{proof}
    The problem SMEM over $P$ is equivalent to SVAL over $P^\polar$, after normalizing the right-hand side of the inequality to $1$ (or to $0$ if $P$ is a cone). If $P$ is well-described then so is $P^\polar$. Thus by equivalence of the four problems over well-described polyhedra with a known feasible point \cite{GLS1988}, we have the desired statement.
\end{proof}

We would like to remark that for a pointed polyhedral cone, generally speaking, the problem SOPT is easy and Theorem \ref{thm:prob_equiv} can be restated by omitting problem \ref{probopt}.

\subsubsection{Extended Formulations}\label{subsec:xf}
Given a polytope $P\subset\mathbb{R}^d$ a polytope $Q\subset\mathbb{R}^{d+r}$ is called an \emph{extended formulation} of $P$ if $P$ is a projection of $Q$. That is, $P=\pi_x(Q)$ where $$\pi_x(Q)=\{x\in\mathbb{R}^d~|~\exists y\in\mathbb{R}^r, (x^\intercal,y^\intercal)^\intercal \in Q\}.$$ 

It may happen that $P$ has exponentially many vertices and facets while an extended formulation of $P$ has polynomially many facets \cite{Kaibel2011ExtendedFormulations, Wolsey2011UsingExtendedFormulations, ConfortiCZ13}. The \emph{extension complexity} of a polytope $P$ -- denoted by $\textrm{xc}(P)$ -- is the minimum number of inequalities required to describe any extended formulation of $P$.


We would like to point out that these notions can be generalized in straightforward way to arbitrary polyhedra.

Having an extended formulation for a polyhedron helps solve the problems SMEM, SVAL, SSEP, and SOPT as follows:

Let $Q=\{(x,y)~|~Ex+Fy\leqslant g\}$ be an extended formulation of $P=\{x~|~Ax\leqslant b\}.$ 

\subsubsection*{Membership testing (SMEM)}
One can check whether $\alpha\in P$ by checking whether the system $Fy\leqslant g-E\alpha$ is feasible or not. If $Q$ is described by polynomially many inequalities then this boils down to solving a polynomial sized linear program.

\subsubsection*{Valid Inequality testing (SVAL)}
An inequality $\alpha^\intercal x\leqslant\beta$ is valid for $P$ if and only if it is valid for $Q$. By duality of Linear Programming $\alpha^\intercal x\leqslant\beta$ is valid for $Q$ if and only if there are non-negative multipliers $\lambda$ such that $\lambda^\intercal E=\alpha, \lambda^\intercal F=0, \lambda^\intercal g\leqslant\beta.$ Thus it can be checked using a linear program of polynomial size.

\subsubsection*{Separation (SSEP)}
Given a point $\alpha$ we see that $\alpha\in P$ if and only if there exists $y$ such that $(\alpha,y)\in Q.$ So if $\alpha\notin P$ then the system $Fy\leqslant g-E\alpha$ is infeasible. By Farkas' Lemma this is equivalent to existence of $\zeta$ such that $\zeta^\intercal F=0, \zeta\geqslant 0, \zeta^\intercal (g-E\alpha) < 0$. On the other hand for any $x\in P$ the system $Fy\leqslant g-Ex$ is feasible. So $\zeta^\intercal (g-Ex) \geqslant 0$. Thus $\zeta^\intercal E x \leqslant \zeta^\intercal g$ is valid for $P$ but does not contain $\alpha$ giving the required separating hyperplane. This hyperplane can be found by solving the linear program $F^\intercal z=0, z\geqslant 0, (g-E\alpha)^\intercal z = -1$.

\subsubsection*{Optimization (SOPT)}
$x^*$ is an optimal solution for the objective function $c^\intercal x$ over $P$ if and only if there exists $y^*$ such that $(x^*,y^*)$ is an optimal solution of $c^\intercal x$ over $Q$. Thus linear optimization over $P$ can be done by linear optimization over $Q$. If $Q$ is defined by a polynomial number of inequalities, then this amounts to solving a linear program of polynomial size.

\subsubsection{(Generalized) Correlation polytope}\label{subsec:cor}
The correlation polytope was introduced by Pitowsky \cite{Pitowsky91} in the context of modeling probabilities that can arise in classical physics. The correlation polytope $\cor(n)$ is defined as $$\cor(n)=\conv{\left\{\left.x\in\{0,1\}^{n+\binom{n}{2}}~\right|~x_{i,j}=x_i\cdot x_j\right\}}.$$

Sometimes the correlation polytope is defined as the convex hull of $xx^\intercal$ for all binary vectors $x$ \cite{FioriniMPTW15}. However, one can see that both definitions are affinely equivalent, with the latter definition containing coordinates $x_{ij}$ as well as $x_{ji}$ of equal value. It is known that the correlation polytope is affinely equivalent to the cut polytope \cite{DESIMONE199071, dezalaurentbook}.

We will use the following known results about $\cor(n)$:
\begin{enumerate}
    \item SMEM for $\cor(n)$ is NP-complete \cite{Pitowsky91}.
    \item SVAL for $\cor(n)$ is not in NP unless NP=coNP \cite{Pitowsky91}.
    \item SOPT for $\cor(n)$ is NP-hard \cite{Karp72, dezalaurentbook}
    \item The extension complexity of $\cor(n)$ is at least $1.5^n$ \cite{FioriniMPTW15, KaibelW15}.
\end{enumerate}

Pitowsky's notion of correlation polytope can be generalized beyond pairwise interactions in a straightforward way by considering higher order correlations. That is, for $k\geqslant 2$, one can define the generalized correlation polytope $\cor^k(n)$ as 

$$\displaystyle\cor^k(n)=\conv{\left\{\left.x\in\{0,1\}^{n+\binom{n}{2}+\cdots+\binom{n}{k}}~\right|~x_{K}=\prod_{i\in K}x_i, \quad\forall K\subseteq N, |K|\leqslant k\right\}}.$$

\begin{obs}\label{obs:gcor}
The vertices of $\cor^k(n)$ are precisely the vectors $w^S \in \BR^{n+\binom{n}{2}+\cdots+\binom{n}{k}}$ for each $S\subseteq N$ such that
\begin{itemize}
    \item $w^S_T=\begin{cases}
        1, & \text{ for } T\subseteq S, 1\leqslant |T|\leqslant k\\
        0, & \text{ otw }
    \end{cases}$.
\end{itemize}
\end{obs}

 Since we will be interested in these generalized correlation polytopes only in the context of establishing NP-hardness, we can use the equivalence between SMEM, SVAL, SSEP, and SOPT for rational polyhedra \cite{GLS1988} even at the cost of using Ellipsoid algorithm. Optimizing over the generalized correlation polytope $\cor^k(n)$ for any fixed $k$ is clearly NP-hard because optimizing over it solves the optimization problem for $\cor(n)$ (cf. Subsection \ref{subsec:xf}) and optimizing over $\cor(n)$ is NP-hard because the max-cut problem is NP-hard \cite{Karp72, dezalaurentbook}. In particular we will use the following results for $\cor^k(n)$ for $k\geqslant 2$ that follow easily from those for $\cor(n)$:
 \begin{enumerate}
    \item SMEM for $\cor^k(n)$ is NP-complete.
    \item SVAL for $\cor^k(n)$ is not in NP unless NP=coNP.
    \item SOPT for $\cor^k(n)$ is NP-hard.
    \item The extension complexity of $\cor^k(n)$ is at least $1.5^n$.
\end{enumerate}

\subsection{Cooperative games}\label{subsec:games}
Let $N:=\{1,\ldots,n\}$ be a finite set of players. A subset of $N$ is called a {\it coalition}, and $2^N$ denotes the set of all coalitions.
A {\it cooperative game with transferable utility} (hereafter abbreviated by {\it game}) (see, e.g., \cite{pesu03,gra16}) is a pair $(N,v)$ where $v:2^N\rightarrow \BR$ is a set function over $N$ such that $v(\varnothing)=0$. The quantity $v(S)$ for some coalition $S$ represents the benefit (or cost reduction) brought by the cooperation of all players in $S$.
When no confusion arises, we will often write $v$ instead of $(N,v)$.

\subsubsection{The Core of a game}
Let $x\in \BR^N$. Usually, $x$ is interpreted as a payment vector, where $x_i$ is the payment for player $i$. For any coalition $S$, we denote by $x(S)$ the total payment given to the coalition $S$, i.e., $x(S)=\sum_{i\in S}x_i$. A payment vector is {\it efficient} if $x(N)=v(N)$, i.e., it redistributes the total benefit earned by cooperation of all players.

For any game $(N,v)$, the {\it core} is the set of efficient payment vectors $x\in \BR^N$ such that every coalition $S$ receives at least its own worth $v(S)$:
\[
C(N,v) = \{x\in\BR^N~\mid~ x(S)\geqslant v(S), \forall S\subseteq N, x(N)=v(N)\}.
\]
If no confusion arises we write $C(v)$ for $C(N,v)$. The core may be empty, meaning that there is no possibility to find a payment satisfying all coalitions, and if not, it is a polytope of dimension at most $n-1$, with at most $2^n-2$ facets.This bound on the number of facets is tight.
In general, the vertices of the core are not known, although they are known for some families of games \cite{grsu16}. Their maximal number is $n!$ \cite{deku02}.

It is well known that supermodular games have a nonempty core, and their vertices are known to be the so-called marginal vectors \cite{Edmonds1970Submodular,sha71}. For a game $v$ and a permutation $\sigma$ on $N$, the associated {\it marginal vector} $x^{v,\sigma}$ is defined by
\begin{equation}\label{eq:mave}
x^{v,\sigma}_i = v(\{j~\mid~ \sigma(j)\leqslant \sigma(i)\}) - v(\{j~\mid~ \sigma(j)<\sigma(i)\}),\quad \forall i\in N,
\end{equation}
where $\sigma(i)$ is the rank of player $i$ in the order induced by $\sigma$.

\subsubsection{The M\"obius coefficients}
Given a game $v$, its {\it M\"obius transform} \cite{rot64} or {\it Harsanyi dividend} \cite{har63} is a set function $m^v$ over $N$ defined by
\[
m^v(S) = \sum_{T\subseteq S}(-1)^{|S\setminus T|}v(T), \quad \forall S\subseteq N.
\]
The inverse relation is given by
\begin{equation}\label{eq:mob}
v(S)  = \sum_{T\subseteq S}m^v(T), \quad \forall S\subseteq N.
\end{equation}
Define the {\it unanimity game} $u_T$ for any $\varnothing\neq T\subseteq N$ by $u_T(S)=1$ if $S\supseteq T$ and 0 otherwise. It is well known \cite{gra16} that $(u_T)_{\varnothing\neq T\subseteq N}$ is a basis of the vector space of games. The coordinates of a game in the basis of unanimity games are precisely the M\"obius coefficients, since (\ref{eq:mob}) can be rewritten as
\[
v = \sum_{T\subseteq N}m^v(T)u_T.
\]

\subsubsection{$k$-additive games}
A game is said to be {\it additive} if for every disjoint coalitions $S,T$ it holds
\[
v(S\cup T)=v(S)+v(T).
\]
Observe that this is equivalent to have for any coalition $S$
\[
v(S) = \sum_{i\in S}v(\{i\}).
\]
A game $v$ is said to be (at most) {\it $k$-additive} for some $1\leqslant k\leqslant n$ \cite{gra96f} if its M\"obius transform vanishes for subsets of more than $k$ elements:
\[
S\subseteq N, |S|>k \Rightarrow m^v(S)=0.
\]

It is easy to see from (\ref{eq:mob}) that $v$ is additive iff $v$ is 1-additive. Observe that a $k$-additive game needs $\sum_{l=1}^k\binom{n}{l}=\Oh{n^k}$ coefficients to be defined, which is polynomial for fixed $k$.

\subsubsection{$l$-monotone games}
A game $v$ is said to be {\it $2$-monotone} or {\it supermodular} or {\it convex} if it satisfies
\[
v(S\cup T)+v(S\cap T) \geqslant v(S) + v(T), \quad \forall S,T\in 2^N.
\]
More generally, for $l\geqslant 2$, a game $v$ is said to be {\it $l$-monotone} if it satisfies
\[
v\Bigl(\bigcup_{i=1}^{l}A_i\Bigr)\geqslant\sum_{\substack{I\subseteq\{1,\ldots,l\}\\ I\neq\varnothing}}(-1)^{|I|+1}\,v\Bigl(\bigcap_{i\in I}A_i\Bigr), \quad \forall A_1,\ldots,A_l\in 2^N.
\]

One can express the supermodularity condition using the M\"obius transform \cite[p.53]{gra16}: $v$ is supermodular iff 
\begin{equation}\label{eq:supermod}
\sum_{S\subseteq T}m^v(\{i,j\}\cup S)\geqslant 0, \quad\forall \{i,j\}\subseteq N, \forall T\subseteq N\setminus\{i,j\}.    
\end{equation}
Observe that for 2-additive games, the above condition just reduces to $m^v(\{i,j\})\geqslant 0$ for all $\{i,j\}\subseteq N$. 

Similarly one can express $l$-monotonicity for any $l\geqslant 2$ using the M\"obius transform \cite[p.53]{gra16}. A game is $l$-monotone for any $l\geqslant 2$ if and only if
\begin{equation}\label{eq:k-monotone}
\sum_{A\subseteq L\subseteq B}m^v(L)\geqslant 0, \quad\forall A,B\subseteq N, A\subseteq B,2\leqslant |A|\leqslant l.
\end{equation}
Observe that $l$-monotonicity implies that $m^v(L)\geqslant 0$ for all $L
\subseteq N$ such that $2\leqslant |L|\leqslant l$. 

\section{The cone of $k$-additive $l$-monotone games}\label{sec:k-l-games}
We will denote by $\binom{X}{k}$ the set of all subsets of a set $X$ that have cardinality $k$. For natural numbers $2\leqslant l\leqslant k$, let us denote the set of $k$-additive $l$-monotone games on $N=[n]$ by $\SM{k}{l}(n).$ 
This is a cone given by the inequalities (\ref{eq:k-monotone}):
\[
\sum_{\substack{A\subseteq L\subseteq B}} m^v(L) \geqslant 0, \quad A\subseteq B\subseteq N,2\leqslant |A|\leqslant l.
\]

Since we are interested in $k$-additive games, the M\"obius coefficients vanish for $|L|>k.$ Moreover, the M\"obius coordinates for singletons are irrelevant for studying this cone, hence we consider hereafter its projection on $\mathbb{R}^{\binom{n}{2}+\cdots +\binom{n}{k}}$, in which case it is a pointed cone. 

Furthermore many of the inequalities are redundant as we show now.

\begin{theorem}\label{thm:klcone}
    Let $2\leqslant l\leqslant k$ be natural numbers. The cone $\SM{k}{l}(n)$ is defined by the following inequalities:
    \begin{eqnarray}
        m^v(A') \geqslant 0, & A'\subseteq N, 2\leqslant |A'|\leqslant l-1 \label{eq:kl_trivial}\\
        \sum_{\substack{T\subseteq B'\\0\leqslant|T|\leqslant k-l}} m^v(A'\cup T)\geqslant 0, & A'\in\binom{N}{l}, B'\subseteq N\setminus A'\label{eq:kl_nontrivial}
    \end{eqnarray}
\end{theorem}
\begin{proof}
   We show that every inequality in (\ref{eq:k-monotone}) is either of type (\ref{eq:kl_trivial}) or ($\ref{eq:kl_nontrivial}$), or it can be obtained by combining these inequalities. Consider a type (\ref{eq:k-monotone}) inequality with sets $A\subseteq B.$ Clearly $|A|\leqslant k$ otherwise $m^v(L)=0$ for all $A\subseteq L\subseteq B$ and the inequality becomes trivial $0\leqslant 0$. 
    
    We will prove by induction on $(|A|,|B|-|A|)$ that $\sum_{A\subseteq L\subseteq B}m^v(L)\geqslant 0$ follows from the inequalities in our list. 

    For the base case, for $|A|=l$ the inequality is a type (\ref{eq:kl_nontrivial}) inequality and for $|B|-|A|=0$ the inequality is a type (\ref{eq:kl_trivial}) inequality.
    
    If $2 \leqslant |A|< l$ and $|B|>|A|$, let $i\in B\setminus A.$ Then, we have 
    \[\begin{aligned}
        &\sum_{A\subseteq L\subseteq B} m^v(L)\\
        =&\sum_{\substack{A\subseteq L\subseteq B\\i\in L}}m^v(L) + \sum_{\substack{A\subseteq L\subseteq B\\i\notin L}}m^v(L) \\
        =&\sum_{\substack{A\cup\{i\}\subseteq L\subseteq B}}m^v(L) + \sum_{\substack{A\subseteq L\subseteq B\setminus\{i\}}}m^v(L)       
    \end{aligned}\]

    Thus, the inequality $\sum_{A\subseteq L\subseteq B}m^v(L)\geqslant 0$ is obtained by adding  $\sum_{\substack{A\cup\{i\}\subseteq L\subseteq B}}m^v(L)\geqslant 0$ and $\sum_{\substack{A\subseteq L\subseteq B\setminus\{i\}}}m^v(L)\geqslant 0$, each of which can inductively be obtained by combining type (\ref{eq:kl_trivial}) and type (\ref{eq:kl_nontrivial}) inequalities.
\end{proof}

\subsection{The trivial case: $l=k$}
For $l=k$ inequalities (\ref{eq:kl_nontrivial}) reduce to the form of (\ref{eq:kl_trivial}) and thus we have the following:

\begin{theorem}\label{thm:k_k_cone}
    The cone $\SM{k}{k}(n)$ is defined by the inequalities $m^v(A)\geqslant 0, \quad \forall A \subseteq N, 2\leqslant|A|\leqslant k.$
\end{theorem}
The extreme rays are just the unit vectors since $\SM{k}{k}(n)$ is the non-negative orthant. Also, SMEM, SVAL, and SSEP can be performed efficiently in time linear in the input size.

\subsection{The easy case: $l=k-1$}\label{subsec:k_k-1}
For $k$-additive games, for the case $l=k-1$, inequalities (\ref{eq:kl_nontrivial}) become 
\begin{equation}\label{eq:10a}
m^v(A)+\sum_{i\in T}m^v(A\cup\{i\}) \geqslant 0, \quad A\in\binom{N}{k-1}, ~T\subseteq N\setminus A.
\end{equation}

Thus $\SM{k}{k-1}(n)$ is the product of the following two cones:
\[ \cS^{k-2}(n):=\left\{\left.m^v\in\BR^{\binom{n}{2}+\cdots+\binom{n}{k-2}}\right| m^v(A)\geqslant 0, \forall A\subseteq N, 2\leqslant |A|\leqslant k-2\right\},\] 

\[\cT^{k-2}(n):=\left\{m^v\in\BR^{\binom{n}{k-1}+\binom{n}{k}}\;\middle|\; m^v \text{ satisfies } (\ref{eq:10a})\right\}.\]

Notice that $\cS^{k-2}(n)$ is the same as $\SM{k-2}{k-2}(n)$ and so its extreme rays are the unit vectors $1^K$ for $K\subseteq N, 2\leqslant |K|\leqslant k-2$, where $1^K$ is the binary vector taking value $1$ as coordinate corresponding to $K$ and $0$ everywhere else. Furthermore if we can characterize the extreme rays of $\cT^{k-2}(n)$ then we have characterized the extreme rays of $\SM{k}{k-1}(n)$ because of the following:

\begin{lemma}\label{lem:cone_product}
    Let $\cC_1\subseteq \BR^{d_1}$ and $\cC_2\subseteq\BR^{d_2}$ be pointed full dimensional polyhedral cones with the set of extreme rays $C_1$ and $C_2$ respectively. Let $\cC=\cC_1\times \cC_2$ be the Cartesian product of the two cones. Then $\cC$ is a pointed polyhedral cone of dimension $d_1+d_2$ whose set of extreme rays is given by $\displaystyle\left\{(x,y)~\left|~ x\in C_1, y=0 \bigvee x=0, y\in C_2\right.\right\}.$
\end{lemma}
\begin{proof}
    It is easy to see that $x\in C_1$ if and only if $(x,0)$ is an extreme ray of $\cC$. Similarly, $y\in C_2$ if and only if $(0,y)$ is an extreme ray of $\cC$. To see that every extreme ray of $\cC$ is of this form, suppose that $(x,y)$ is an extreme ray of $\cC$ such that $x\neq 0, y\neq 0$. Then $x\in\cC_1$ and $y\in\cC_2.$ So $(x,0)\in\cC, (0,y)\in\cC$. Therefore $(x,y)$ is obtained as a conic combination of other points of $\cC$ and thus cannot be an extreme ray of $\cC.$
\end{proof}

\begin{theorem}\label{thm:setT}
The set $\cT^{k-2}(n)$ is a pointed polyhedral cone of dimension $d=\binom{n}{k-1}+\binom{n}{k}$ whose extreme rays are:
\begin{itemize}
\item $r^K=1^K$, \hfill $K\subseteq N$, $k-1\leqslant |K|\leqslant k$
\item $\rho^K = -1^K+\sum_{i\in K}1^{K\setminus\{i\}}$,  \hfill $K\in\binom{N}{k}$,
\end{itemize}
where $1^S$ takes value 1 for coordinate corresponding to $m^v(S)$ and 0
otherwise. 
Hence, the number of extremal rays is $d+\binom{n}{k}$.
\end{theorem}
\begin{proof}
    It is easy to see that $\cT^{k-2}(n)$ is a pointed cone. Since $r^K$ is feasible for each $K\subseteq N, k-1\leqslant|K| \leqslant k$, the dimension follows.
    
    Let $K\subseteq N$, $k-1\leqslant|K|\leqslant k$. We claim that the set of tight inequalities of $\cT^{k-2}(n)$ that correspond to $r^K$ are all the inequalities that do not contain $m^v(K)$. We will argue that $m^v(L)=0$ is implied for all $L\neq K$ by these tight inequalities and thus only $m^v(K)$ remains undetermined. So the system has rank $d-1$ and thus, $r^K$ is an extreme ray. To this end, let $L\neq K$. We identify the following cases:
    \begin{itemize}
        \item Suppose $|L|=k-1$.\\ Then $m^v(L)=0$ is obtained by taking $A=L$ and $T=\emptyset.$ 
        \item Suppose $|L|=k, |K|=k-1, K\subseteq L$.\\
        Then, take $i\in L\setminus K$ and pick $j\in K$. Taking $A=(K\setminus\{j\})\cup\{i\}$ and $T=\emptyset$ we get $m^v((K\setminus\{j\})\cup\{i\})=0$. Taking $A=(K\setminus\{j\})\cup\{i\}$ and $T=\{j\}$ we get $m^v((K\setminus\{j\})\cup\{i\})+m^v(L)=0$. Thus $m^v(L)=0.$
        \item Suppose $|L|=k, K\nsubseteq L$.\\ Then let $i\in L\setminus K$ be arbitrary. Then $m^v(L)=0$ is implied by $m^v(L\setminus\{i\})=0$, and $m^v(L\setminus\{i\})+m^v(L)=0$. The former tight inequality is obtained by taking $A=L\setminus\{i\}, T=\emptyset$ and the latter by taking $A=L\setminus\{i\}, T=\{i\}.$ 
    \end{itemize}

    Now, let $K\in\binom{N}{k}$. We will prove that $\rho^K$ is an extreme ray. Consider the following set of tight inequalities:
    
    \[\begin{aligned}
        m^v(L)=0, & \quad |L|=k-1, L\nsubseteq K\\
        m^v(L\setminus\{i\})+m^v(L)=0, & \quad |L|=k, L\neq K, i\in L\setminus K, |L\setminus K|>1\\
        m^v(L\setminus\{i\}) + m^v(L) = 0 & \quad |L|=k, L\neq K, i\in L\cap K, |L\setminus K|=1.
    \end{aligned}\]
The first set of equations comes from (\ref{eq:10a}) with $A=L$ and $T=\varnothing$, and determines $m^v(L)=0$ for all $L\subseteq N$ with $|L|=k-1$ and $L\neq K\setminus \{i\}$ for some $i$. The second set of equations comes from (\ref{eq:10a}) with $A=L\setminus\{i\}$, and $T=\{i\}$. As $L\setminus\{i\}\not\subseteq K$, we get $m^v(L\setminus \{i\})=0$  from above and therefore $m^v(L)=0$ for all $L$ s.t. $|L|=k$, $L\neq K$ and $|L\setminus K|>1$. The third set of equations comes from (\ref{eq:10a}) with $A=L\setminus \{i\}$ and $T=\{i\}$. As $L\setminus\{i\}\not\subseteq K$, we get $m^v(L\setminus \{i\})=0$  and therefore $m^v(L)=0$ for all $L$ s.t. $|L|=k$, $L\neq K$ and $|L\setminus K|=1$. In summary, $m^v(L)=0$ for all $L$ of cardinality $k$ or $k-1$, and different from $K$ and $K\setminus \{i\}$, $i\in K$. The remaining equalities are
    \[
        m^v(K\setminus\{i\})+m^v(K)=0, \quad i \in K.
    \]

    This is a system of $k$ equations with $k+1$ unknowns leaving exactly one variable undetermined. Thus the rank of the whole system of tight inequalities is $d-1$ and $\rho^K$ is the unique solution (up to a scalar multiple). Thus $\rho^K$ is an extreme ray.

    Finally, we will prove that no other extreme rays exist. To this aim, we give a conic decomposition of any ray $w\in \cT^{k-2}(n)$ in terms of the $r^K$'s and $\rho^K$'s. For any $K\in\binom{N}{k}$, let us define the positive and negative components of $w$:
    \[
    \alpha_K:=\max(0,w_K), \quad \beta_K:=\max(0,-w_K).
    \]
    which are clearly nonnegative.
    Now, for $A\in\binom{N}{k-1}$ let us define
    \[
    \alpha_A = w_A - \sum_{i\not\in A}\beta_{A\cup\{i\}}.
    \]
    These coefficients are nonnegative. Indeed, letting $T_A:=\{i\not\in A\left| w_{A\cup\{i\}}<0\right.\}$, we get
    \[
    \alpha_A = w_A+\sum_{i\in T_A}w_{A\cup\{i\}}
    \]
which is nonnegative as it is a particular instance of (\ref{eq:10a}). Finally, we have
\[
w = \sum_{S:k-1\leqslant |S|\leqslant k}\alpha_Sr^S + \sum_{K\in\binom{N}{k}}\beta_K\rho^K
\]
coordinatewise. Indeed, for any $K\in\binom{N}{k}$, it holds
\[
w_K=\alpha_K+\beta_K(-1) = w_K
\]
  while for $A\in\binom{N}{k-1}$ we have
  \[
w_A = \alpha_A + \sum_{i\not\in A}\beta_{A\cup i}(-0+1) = w_A.
  \]
\end{proof}

Combining our earlier observation that $\SM{k}{k-1}(n)$ is the product of $\cS^{k-2}(n)$ and $\cT^{k-2}(n),$ and that $\cS^{k-2}(n)$ is the nonnegative orthant, together with Lemma \ref{lem:cone_product}
and Theorem \ref{thm:setT} we get the following:

\begin{corollary}\label{cor:k_k-1_extrays}
The set $\SM{k}{k-1}(n)$ is a pointed polyhedral cone of dimension $d=\binom{n}{2}+\cdots+\binom{n}{k}$ whose extreme rays are:
\begin{itemize}
\item $r^K=1^K$, $K\subseteq N$, $2\leqslant |K|\leqslant k$
\item $\rho^K = -1^K+\sum_{i\in K}1^{K\setminus\{i\}}$,  $K\in\binom{N}{k}$,
\end{itemize}
where $1^S$ takes value 1 for coordinate corresponding to $m^v(S)$ and 0
otherwise. 
Hence, the number of extremal rays is $d+\binom{n}{k}$.
    
\end{corollary}

Observe that if $P\subseteq\BR^n$ is the conic hull of $\{v_1,\ldots,v_m\}$ then it is also a projection of a polyhedral cone with at most $m$ facets, namely $\left\{(x,\lambda)\in\BR^n\times\BR^m~\left|~x=\sum_{i=1}^m \lambda_iv_i,\lambda\geqslant 0\right.\right\}.$ Thus we have the following:

\begin{theorem}\label{thm:k_k-1_xc}
\(\xc{\SM{k}{k-1}(n)}=O(n^k)\). 
\end{theorem}

An algorithmic consequence of the above is that one can check whether a $k$-additive game on $N=[n]$, given by its M\"obius coefficients, is $(k-1)$-monotone, by solving a linear program whose size is linear in the input size. One can also solve SVAL and SSEP efficiently.

\subsection{The hard case: $l\leqslant k-2$}\label{subsec:k_k-2}
As before, we denote the set of $k$-additive $l$-monotone games on $N=[n]$ by $\SM{k}{l}(n).$ We now show that for $2\leqslant l\leqslant k-2$, this cone has a complicated structure. In particular, the generalized correlation polytope $\cor^{k-l}(n-l)$ is hidden in the facial structure of $(-\SM{k}{l}(n))^\polar$ in a way that allows us to propagate oracles for membership, optimization, and valid inequality testing for this cone to that of the generalized correlation polytope. Using simple observations about extension complexity we are able to translate lower bounds for the generalized correlation polytope to that of $\SM{k}{l}(n).$

The geometric core of the argument is a facial realization of the correlation polytope inside the dual. More precisely, we exhibit a face \(F\) of the dual cone \((-\SM{k}{l}(n))^\polar\), choose an affine hyperplane \(H\) such that \(F\cap H\) is bounded and encodes multiple copies of the correlation polytope, and then take a face of \(F\cap H\) that is affinely equivalent to one of those copies. The complexity consequences are then derived by translating the facial embedding into decision problems over the original cone and its dual description. 

\subsubsection*{To the Dual}

By Theorem \ref{thm:klcone}, the cone $\SM{k}{l}(n)$ - the set of $k$-additive $l$-monotone games on $N=[n]$, is given by the inequalities:
\begin{eqnarray*}
    m^v(A) \geqslant 0, & A\subseteq N, 2\leqslant |A|\leqslant l-1\\
    \sum_{t=0}^{k-l}\sum_{T\in \binom{S}{t}} m^v(A\cup T)\geqslant 0, & A\in \binom{N}{l}, S\subseteq N\setminus A
\end{eqnarray*}

By a change of variables and setting $x_K=-m^v(K)$, we get the affinely isomorphic cone $-\SM{k}{l}(n)$ defined by
\begin{eqnarray*}
    x_A \leqslant 0, & \quad A\subseteq N, 2\leqslant |A| \leqslant l-1,\\
    \sum_{t=0}^{k-l} \sum_{T\in\binom{S}{t}} x_{A\cup T}  \leqslant 0, & \quad A\in\binom{N}{l}, S \subseteq N\setminus A.
\end{eqnarray*}

Then the polar dual $(-\SM{k}{l}(n))^\polar$ is the conic hull of the vectors $w^A\in\BR^{\binom{n}{2}+\cdots+\binom{n}{k}}$ defined for each $A\subseteq N, 2\leqslant|A|\leqslant l-1$ and vectors $w^{A,S}\in\BR^{\binom{n}{2}+\cdots+\binom{n}{k}}$ defined for each  $A\in\binom{N}{l}, S\subseteq N\setminus A$ as follows:
\begin{itemize}
    \item $w^A_K=\begin{cases}
        1, & \text{ if }K=A\\
        0, & \text{ otw }
    \end{cases}$
    \item $w^{A,S}_K=\begin{cases}
        1, & \text{ if } K=A\cup T, T\subseteq S, |T|\leqslant k-l \\
        0, & \text{ otw }
    \end{cases}$
\end{itemize}

The next result follows easily from Theorem \ref{thm:prob_equiv}.

\begin{lemma}\label{lem:2polar}
    If SMEM for $\SM{k}{l}(n)$ is in NP then so is SVAL for $(-\SM{k}{l}(n))^\polar.$ If SMEM for $(-\SM{k}{l}(n))^\polar$ is NP-complete then so is SVAL for $\SM{k}{l}(n).$
\end{lemma}

\subsubsection*{A face of the dual}
\begin{lemma}
    Let $\displaystyle \cA\cC^k_l(n)=(-\SM{k}{l}(n))^\polar\bigcap_{\substack{K\subseteq N\\ 2\leqslant|K|\leqslant l-1}}\left\{\left.x\in\BR^{\binom{n}{2}+\cdots+\binom{n}{k}}~\right|~x_K=0\right\}.$ Then $\cA\cC^k_l(n)$ is the conic hull of the following vectors $w^{A,S}\in\BR^{\binom{n}{2}+\cdots+\binom{n}{k}}$ for each $A\in\binom{N}{l}, S\subseteq N\setminus A$:
    \begin{itemize}
        \item $w^{A,S}_K=\begin{cases}
            1, & \text{ if } K=A\cup T, T\subseteq S, |T|\leqslant k-l \\
            0, & \text{ otw }
        \end{cases}$
    \end{itemize}
\end{lemma}

\begin{proof}
    $x_K\geqslant 0$ is valid for $(-\SM{k}{l}(n))^\polar$ for every $K\subseteq N, 2\leqslant |K|\leqslant k$. So $\cA\cC^k_l(n)$ is the face containing exactly those extreme rays that do not have $w_K=1$ for any $K\subseteq N, 2\leqslant |K|\leqslant l-1.$
\end{proof}

\begin{lemma}\label{lem:polar2face}
    If SVAL for $(-\SM{k}{l}(n))^\polar$ is in NP then so is SVAL for $\cA\cC^k_l(n).$ If SMEM for $\cA\cC^k_l(n)$ is NP-complete then so is SMEM for $(-\SM{k}{l}(n))^\polar.$    
\end{lemma}
\begin{proof}
    Suppose that SVAL for $(-\SM{k}{l}(n))^\polar$ is in NP. Given an inequality $\alpha^\intercal x\leqslant 0$ if it is valid for $\cA\cC^k_l(n)$ we want to certify it using a polynomial time verifiable certificate. Consider the inequality $\alpha'^\intercal x\leqslant 0$ with 
    $$\alpha'_K=\begin{cases}
        \alpha_K, & \quad |K|\geqslant l,\\
        -\displaystyle\sum_{|L|\geqslant l}|\alpha_L| & \quad 2\leqslant |K| \leqslant l-1
    \end{cases}$$

    Consider any extreme ray $w$ of $(-\SM{k}{l}(n))^\polar$. Without loss of generality we assume that $w_K\in\{0,1\}$ for all $K\subseteq N.$ If $w$ does not lie in the face $\cA\cC^k_l(n)$ then there exists $K\subseteq N, 2\leqslant |K| \leqslant l-1$ such that $w_K=1.$ So $\alpha'^\intercal w \leqslant 0.$ If $w$ is in the face $\cA\cC^k_l(n)$ then $\alpha^\intercal w=\alpha'^\intercal w.$ Thus $\alpha'^\intercal x\leqslant 0$ is valid for ${(-\SM{k}{l}(n))^\polar}$ if and only if $\alpha^\intercal x\leqslant 0$ is valid for ${\cA\cC^k_l(n)}.$ Note that the input size of $\alpha'$ is at most size of $\alpha$ plus $\binom{n}{l}+\cdots+\binom{n}{k}$ and so one can verify, in polynomial time, that $\alpha'$ is indeed produced from $\alpha$. This check together with a certificate of validity of $\alpha'^\intercal x\leqslant 0$ for $(-\SM{k}{l}(n))^\polar$ gives a certificate for validity of $\alpha^\intercal x\leqslant 0$ for $\cA\cC^k_l(n).$

    Membership in $\cA\cC^k_l(n)$ is equivalent to membership in $(-\SM{k}{l}(n))^\polar$ together with satisfying $x_K=0$ for all $K\subseteq N, 2\leqslant|K|\leqslant l-1$. Thus given a polynomial time algorithm for SMEM over $(-\SM{k}{l}(n))^\polar$ one can solve SMEM over $\cA\cC^k_l(n)$ in polynomial time. Thus SMEM over $(-\SM{k}{l}(n))^\polar$ is NP-hard if SMEM over $\cA\cC^k_l(n)$ is NP-complete.

    To see that SMEM over $(-\SM{k}{l}(n))^\polar$ is in NP, if $x\in(-\SM{k}{l}(n))^\polar$ then by Carath\'eodory's theorem there exist polynomially many extreme rays of $(-\SM{k}{l}(n))^\polar$ whose conic hull contains $x$. Given those extreme rays this is a linear programming task and thus polynomial time solvable. Thus membership in $(-\SM{k}{l}(n))^\polar$ can be certified with a polynomially checkable proof.
\end{proof}

\subsubsection*{Slicing the face to make a polytope}
The cone $\cA\cC^k_l(n)$ is embedded in $\BR^{\binom{n}{2}+\cdots+\binom{n}{k}}$, however for every coordinate corresponding to a subset $K$ with $2\leqslant |K| \leqslant l-1$ every point in $\cA\cC^k_l(n)$ has value zero. So from now on we will identify $\cA\cC^k_l(n)$ with the affinely equivalent cone in $\BR^{\binom{n}{l}+\cdots+\binom{n}{k}}$ obtained by dropping these ``zero-coordinates''.

\begin{lemma}
Let $\displaystyle\cB\cC^k_l(n)=\cA\cC^k_l(n)\bigcap\left\{x\in\BR^{\binom{n}{l}+\cdots+\binom{n}{k}}~\left|~\sum_{K\in \binom{N}{l}}x_K=1\right.\right\}$. Then $\cB\cC^k_l(n)$ is the convex hull of the following vectors $w^{A,S}\in\BR^{\binom{n}{l}+\cdots+\binom{n}{k}}$ for each $A\in\binom{N}{l}, S\subseteq N\setminus A$:
    \begin{itemize}
        \item $w^{A,S}_K=\begin{cases}
            1, & \text{ for } K=A\cup T, ~ T\subseteq S, ~|T|\leqslant k-l \\
            0, & \text{ otw }
        \end{cases}$ 
    \end{itemize}
\end{lemma}
\begin{proof}
    $\sum_{K\in \binom{N}{l}}x_K\geqslant 0$ is valid for $\cA\cC^k_l(n)$ with every extreme ray of the cone satisfying the inequality strictly. So $\sum_{K\in \binom{N}{l}}x_K=1$ intersects every extreme ray and these intersections define the vertices of the resulting polytope.
\end{proof}

\begin{lemma}\label{lem:face2polytope}
    If SVAL over $\cA\cC^k_l(n)$ is in NP then so is SVAL over $\cB\cC^k_l(n).$ If SMEM over $\cB\cC^k_l(n)$ is NP-complete then so is SMEM over $\cA\cC^k_l(n).$
\end{lemma}
\begin{proof}
    An inequality $\alpha^\intercal x\leqslant \beta$ is valid for $\cB\cC^k_l(n)$ if and only if the inequality $$\displaystyle\alpha^\intercal x-\beta\sum_{K\in\binom{N}{l}}x_K\leqslant 0$$ is valid for $\cA\cC^k_l(n).$ Thus, a proof of validity of $\displaystyle\alpha^\intercal x-\beta\sum_{K\in\binom{N}{l}}x_K\leqslant 0$ for  $\cA\cC^k_l(n)$ can be used as a proof of validity of $\alpha^\intercal x\leqslant\beta$ for $\cB\cC^k_l(n)$.

    Now, given a point $x$ if we wish to check whether $x\in\cB\cC^k_l(n)$ then we can check that $x \in \cA\cC^k_l(n)$ and $x$ lies on $\sum_{K\in\binom{N}{l}}x_K=1.$ Thus if SMEM over $\cB\cC^k_l(n)$ is NP-complete then SMEM over $\cA\cC^k_l(n)$ is NP-hard.  To see that membership in $\cA\cC^k_l(n)$ is in NP, we again observe that membership of $x$ can be certified using a polynomially many extreme rays of $\cA\cC^k_l(n)$ that contain $x$ in the conic hull - a linear programming problem given the extreme points.
\end{proof}

\subsubsection*{Isolating the correlation polytope}
\begin{lemma}
Let $\displaystyle\cC\cC^k_l(n)=\cB\cC^k_l(n)\bigcap\left\{\left.x\in\BR^{\binom{n}{l}+\cdots+\binom{n}{k}}~\right|~x_K=1, K=\{n-l+1,\ldots,n\}\right\}.$ Then $\cC\cC^k_l(n)$ is affinely isomorphic to the generalized correlation polytope $\cor^{k-l}(n-l).$
\end{lemma}
\begin{proof}
    The inequality $x_K\leqslant 1$ is valid for $\cB\cC^k_l(n)$ for every $K\subseteq N, l\leqslant |K|\leqslant k.$ Thus, $\cC\cC^k_l(n)$ is a face of $\cB\cC^k_l(n).$ This face contains exactly the vertices $w^{\{n-l+1,\ldots,n\},S}$ of $\cB\cC^k_l(n)$ for $S\subseteq \{1,\ldots,n-l\}.$

    Dropping every coordinate other than $x_{K}$ for $K\supsetneq \{n-l+1,\ldots,n\}$, and denoting the resulting point by $w^S$ we obtain a vertex exactly for each $S\subseteq\{1,\ldots,n-l\}$ with 
    \begin{itemize}
        \item $w^S_{\{n-l+1,\ldots,n\}\cup T}=\begin{cases}
            1, & \text{ if } T\subseteq S, ~ 1\leqslant |T|\leqslant k-l\\
            0, & \text{ otw }
        \end{cases}$.
    \end{itemize}

    Identifying the coordinates $x_{\{n-l+1,\ldots,n\}\cup K}$ with $x_K$ for each $K\subseteq [n-l]$ we see that $w^S$ are precisely the vertices of $\cor^{k-l}(n-l)$ for $S\subseteq [n-l]$ by Observation \ref{obs:gcor}.
\end{proof}

\begin{lemma}\label{lem:faceofpolytope}
    If SVAL over $\cB\cC^k_l(n)$ is in NP then so is SVAL over $\cC\cC^k_l(n).$ If SMEM over $\cC\cC^k_l(n)$ is NP-complete then so is SMEM over $\cB\cC^k_l(n).$
\end{lemma}
\begin{proof}
    Suppose that SVAL for $\cB\cC^k_l(n)$ is in NP. Given an inequality $\alpha^\intercal x\leqslant\beta$, if it is valid for $\cC\cC^k_l(n)$, we want to certify it using a polynomial time verifiable certificate. Consider the inequality $\alpha'^\intercal x \leqslant \beta$ with
    \[
    \alpha'_K=\begin{cases}
        \alpha_K, & \quad K\notin\binom{N}{l} \vee K=\{n-l+1,\ldots,n\}\\
        \displaystyle\left(-\sum_{L\notin\binom{N}{l}}|\alpha_L|\right)+\beta, & \quad K\in\binom{N}{l} \wedge K\neq\{n-l+1,\ldots,n\}
    \end{cases}
    \]

    Consider any vertex $v$ of $\cB\cC^k_l(n)$. If $v$ is not a vertex of $\cC\cC^k_l(n)$, then there exists $K\in\binom{N}{l}$ such that $K\neq\{n-l+1,\ldots,n\}$ and $x_K=1$. So $\alpha'^\intercal v\leqslant \beta$. If $v\in\cC\cC^k_l(n)$ then $\alpha'^\intercal v=\alpha^\intercal v\leqslant \beta.$ Thus $\alpha^\intercal x\leqslant \beta$ is valid for $\cC\cC^k_l(n)$ if and only if $\alpha'^\intercal x\leqslant\beta$ is valid for $\cB\cC^k_l(n).$ Thus a certificate for validity of $\alpha'^\intercal x\leqslant\beta$ for $\cB\cC^k_l(n)$ can be used as a certificate for validity of $\alpha^\intercal x\leqslant \beta$ for $\cC\cC^k_l(n).$

    Now suppose SMEM over $\cC\cC^k_l(n)$ is NP-complete. A point $x\in \cC\cC^k_l(n)$ if and only if $x\in\cB\cC^k_l(n)$ and $x_{\{n-l+1,\ldots,n\}}=1.$ The latter check can obviously be done in polynomial time. So SMEM over $\cB\cC^k_l(n)$ is NP-hard. To see that it is in NP we again observe that membership in $\cB\cC^k_l(n)$ be certified by producing polynomially many vertices of $\cB\cC^k_l(n)$ whose convex hull contains $x.$ 
\end{proof}

\begin{theorem}\label{thm:hardness}
    Let $k,l$ be positive integral constants such that $2\leqslant l \leqslant k -2$. Then the following holds:
    \begin{itemize}
        \item SMEM over $\SM{k}{l}(n)$ is not in NP unless NP=coNP.
        \item SVAL over $\SM{k}{l}(n)$ is NP-complete.
        \item SSEP over $\SM{k}{l}(n)$ is NP-hard.
        \item $\xc{\SM{k}{l}(n)}\geqslant 1.5^{n-l}$
    \end{itemize}
\end{theorem}
\begin{proof}
Suppose NP$\neq$coNP. Then SVAL over $\cor^{k-l}(n-l)$ is not in NP. Since $\cor^{k-l}(n-l)$ is the same as $\cC\cC^k_l(n)$ with a simple embedding in $\BR^{\binom{n}{l}+\cdots+\binom{n}{k}}$ SVAL over $\cC\cC^k_l(n)$ is not in NP. By Lemma \ref{lem:faceofpolytope} SVAL over $\cB\cC^k_l(n)$ is not in NP. Then by Lemma \ref{lem:face2polytope} SVAL over $\cA\cC^k_l(n)$ is not in NP. Then, by Lemma \ref{lem:polar2face}, SVAL over $(-\SM{k}{l}(n))^\polar$ is not in NP. Finally, by Lemma \ref{lem:2polar}, SMEM over $\SM{k}{l}(n)$ is not in NP.

The same chain of implications starting with NP-completeness of SMEM over $\cor^{k-l}(n-l)$ we get that SVAL over $\SM{k}{l}(n)$ is NP-complete. Since the NP-hardness of SMEM over $\cor^{k-l}(n-l)$ is derived from that of $\cor(n)$ via the equivalence of SMEM, SVAL, SSEP, and SOPT \cite{GLS1988}, this NP-completeness is with respect to polynomial-time Turing reductions. NP-hardness of SSEP follows from polynomial equivalence of SSEP, SVAL, and SMEM together with NP-completeness of SVAL.

For the extension complexity lower bound, note that the extension complexity of a pointed polyhedral cone is the same as that of its polar dual. So $\xc{(-\SM{k}{l}(n))^\polar}=\xc{\SM{k}{l}(n)}.$ Since the extension complexity of a face is upper bounded by the extension complexity of the whole polyhedron \cite{AvisT15}, we have $\xc{(-\SM{k}{l}(n))^\polar}\geqslant\xc{\cA\cC^k_l(n)}$. It is easy to see that slicing a pointed cone with a hyperplane that intersects every extreme ray preserves extension complexity so $\xc{\cA\cC^k_l(n)}=\xc{\cB\cC^k_l(n)}.$ Once again taking a face we obtain $\xc{\cB\cC^k_l(n)}\geqslant\xc{\cC\cC^k_l(n)}.$ Finally since $\cC\cC^k_l(n)$ is affinely isomorphic to $\cor^{k-l}(n-l)$ we get that $\xc{\cC\cC^k_l(n)}\geqslant 1.5^{n-l}.$ Putting it all together we obtain,
\[
\xc{(-\SM{k}{l}(n))^\polar}=\xc{\SM{k}{l}(n)}\geqslant\xc{\cA\cC^k_l(n)}=\xc{\cB\cC^k_l(n)}\geqslant\xc{\cC\cC^k_l(n)}\geqslant 1.5^{n-l}.
\]
\end{proof}

We would like to again remark that no good characterization of the extreme rays of the cone $\SM{k}{l}(n)$ is known for $l\leqslant k-2$. Theorem \ref{thm:hardness} gives a theoretical basis for this phenomenon. By Carath\'eodory's theorem, membership in $\SM{k}{l}(n)$ can be certified by giving $d$ extreme rays -- where $d$ is the dimension of the cone -- and checking that the given point can be written as a conic combination of these rays. Thus an efficiently checkable characterization of the extreme rays would put membership checking -- SMEM -- over $\SM{k}{l}(n)$ in the class NP and would thus imply NP=coNP. We further note that a stronger coNP-completeness for SMEM can be proved using our methods. However, as stated currently our result has the immediate interpretation in terms of why good characterization for the extreme rays of monotonicity cones are lacking, so we state it in the current form.



\section{The core of \(k\)-additive \(l\)-monotone games}\label{sec:core}


For $l\geqslant 2$, $l$-monotonicity is a property stronger than supermodularity. That is, every $l$-monotone game for any $l\geqslant 2$ is supermodular. Therefore the vertices of the core are described by permutations of $[n].$ Let $\sigma$ be a permutation, then the core has a corresponding vertex $x^\sigma:=x^{v,\sigma}$ given by (\ref{eq:mave}). Furthermore every vertex corresponds to a permutation.

In this section, we find extended formulations for the cases $l=k$ and $l=k-1$.

\subsection{The case $l=k$}
Given a game $(N,v)$ and positive integer $k$ define the polytope $Q_k(N,v)$ by the following equations and inequalities:
            \begin{align}
                &x_i=\sum_{\substack{\emptyset\neq J\not\ni i\\ |J|\leqslant k-1}}m^v(\{i\}\cup J)y_{i,J}+m^v(\{i\}), & i\in N\label{eq:kl-core-1}\\
                &y_{i,J}\geqslant 0, & \emptyset\neq J\subset N\setminus\{i\}, {|J|\leqslant k-1},i\in N\label{eq:kl-core-2}\\
                &\sum_{i\in T}y_{i, T\setminus\{i\}}=1& T\subseteq N, 2\leqslant |T|\leqslant k\label{eq:kl-core-3}
            \end{align}
\begin{theorem}\label{thm:corexc}
    Let $(N,v)$ be a $k$-additive $k$-monotone game. Then $C(N,v)=\Pi_x(Q_k(N,v)),$ that is, the core of the game $(N,v)$ is a projection of $Q_k(N,v)$ onto the $x$ coordinates.    
\end{theorem}
\begin{proof}
    Let $x^\sigma$ be a vertex of $C(N,v)$ corresponding to some permutation $\sigma$ of $N.$ Define 
    \[y^\sigma_{i,J}:=\begin{cases}
        1, & \text{if } \sigma(J)<\sigma (i)\\
        0, & \text{otherwise}
    \end{cases},\]
    where $\sigma(J)<\sigma(i)$ means that $\sigma(j)<\sigma(i),~ \forall j\in J.$ 
    
    By using (\ref{eq:kl-core-1}), we get
    \[
    x^\sigma_i=\sum_{\substack{\emptyset\neq J\not\ni i\\|J|\leqslant k-1}}m^v(\{i\}\cup J)y^\sigma_{i,J}+m^v(\{i\}), \forall i\in N
    \]
    and so $(x^\sigma,y^\sigma)\in Q_k(N,v).$ This proves that $C(N,v)\subseteq\Pi_x(Q_k(N,v)).$

    For the converse containment $\Pi_x(Q_k(N,v))\subseteq C(N,v)$ we will show that all the inequalities defining $C(N,v)$ are implied by the inequalities of $Q_k(N,v).$

    Let us first consider the equality $x(N)=v(N).$ We have 
    \[\begin{aligned}
        x(N)=&\sum_{i\in N}x_i\\
        =&\sum_{i\in N}\left(\sum_{\substack{\emptyset\neq J\not\ni i\\|J|\leqslant k-1}}m^v(\{i\}\cup J)y_{i,J}+m^v(\{i\})\right)\\
        =&\sum_{2\leqslant |K| \leqslant k}m^v(K)\left(\sum_{i\in K}y_{i,K\setminus\{i\}}\right)+\sum_{i\in N}m^v(\{i\})\\
        =&\sum_{\substack{T\subseteq N\\|T|\leqslant k}}m^v(T)\\
        =&v(N).
    \end{aligned}\]

    Now consider any subset $S\subseteq N$ and the corresponding core inequality $x(S)\geq v(S).$

    We have
    \[\begin{aligned}
        x(S)=&\sum_{i\in S}x_i\\
        =&\sum_{i\in S}\left(\sum_{\substack{\emptyset\neq J\not\ni i\\|J|\leqslant k-1}}m^v(\{i\}\cup J)y_{i,J}+m^v(\{i\})\right)\\
        =&\sum_{\substack{K\subseteq S\\2\leqslant|K|\leqslant k}}m^v(K)\left(\sum_{i\in K}y_{i,K\setminus\{i\}}\right)+\sum_{\substack{i\in S\\\emptyset\neq J\not\subseteq S\\J\not\ni i,|J|\leqslant k-1}}m^v(\{i\}\cup J)y_{i,J}+\sum_{i\in S}m^v(\{i\})\\
        =&\sum_{\substack{K\subseteq S\\2\leqslant|K|\leqslant k}}m^v(K)+\sum_{i\in S}m^v(\{i\})+\sum_{\substack{i\in S\\\emptyset\neq J\not\subseteq S\\J\not\ni i,|J|\leqslant k-1}}m^v(\{i\}\cup J)y_{i,J}\\
        =&v(S)+\sum_{\substack{i\in S\\\emptyset\neq J\not\subseteq S\\J\not\ni i,|J|\leqslant k-1}}m^v(\{i\}\cup J)y_{i,J}
    \end{aligned}\]

Since $v$ is $k$-monotone and $k$-additive, it follows that $m^v(\{i\}\cup J)\geqslant 0$ in every term of the above sum.    
Therefore $x(S)-v(S)\geqslant 0$ is obtained from non-negative combination of $y_{i,J}\geqslant 0$ for $i\in S, J\not\subseteq S, J\not\ni i$ and the equalities defining $Q_k(N,v).$
\end{proof}

The extended formulation thus obtained is explicitly constructed and yields the following algorithmic corollary.

\begin{corollary}
    Let $(N,v)$ be a $k$-additive $k$-monotone game. Then, there exists a Linear Program with $O(n^k)$ variables and constraints for SMEM, SVAL, SSEP, and SOPT over $\core(N,v)$.
\end{corollary}

This point is important conceptually. Earlier tractability results can be interpreted through separation or oracle-based optimization, whereas our formulation gives direct access to the core as a small explicit lifted polytope. 

\subsection{The case $l=k-1$}\label{subsec:core_k_k-1}
Based on the extremal ray decompositions of $\SM{k}{k-1}(n)$, we can derive an extended formulation for the core. Recall that $\SM{k}{k-1}(n) = \cS^{k-2}(n)\times\cT^{k-2}(n)$. The extremal rays of $\cT^{k-2}(n)$ are the vectors $r^K$, $k-1\leqslant|K|\leqslant k$, which correspond to the unanimity games $u_K$, and the vectors $\rho^K$, $|K|=k$, which correspond to the games
\[
g_K = -u_K + \sum_{i\in K}u_{K\setminus \{i\}}.
\]
The extremal rays of $\cS^{k-2}(n)$ are simply the unanimity games $u_K$ for $2\leqslant |K|\leqslant k-2$. 
It follows that any game $v\in\SM{k}{k-1}(n)$ can be written as
\[
v = v_0 + \sum_{\substack{K\subseteq N\\2\leqslant|K|\leqslant k}}\alpha_Ku_K + \sum_{K\in\binom{N}{k}}\beta_Kg_K,
\]
where $v_0$ is the additive part of $v$, i.e., $v_0 = \sum_{i\in N}m^v(\{i\})u_{\{i\}}$, and all coefficients $\alpha_K,\beta_K$ are nonnegative. These coefficients can be computed explicitly from the M\"obius coefficients of $v$; see the proof of Theorem~\ref{thm:setT}. 

As $l\geqslant 2$, all games in the decomposition are supermodular. It follows from the properties of positive homogeneity of the core and of Minkowski additivity of the core w.r.t. supermodular games, i.e., $C(v_1+v_2)=C(v_1)+C(v_2)$ for any supermodular games $v_1,v_2$ \cite{sha71, gra16, TB02}, that we obtain
\[
C(v) = C(v_0 ) + \sum_{\substack{K\subseteq N\\2\leqslant |K|\leqslant k}}\alpha_K C(u_K) + \sum_{K\in\binom{N}{k}}\beta_KC(g_K).
\]
The core of an additive game is reduced to a singleton, namely $C(v_0) = \{(v(\{i\}))_{i\in N}\}$. The core of the unanimity game $u_T$ is easily seen to be the simplex $\Delta_T = \{x\in\BR^n_+\mid \sum_{i\in T}x_i=1, x_i=0 \text{ for } i\notin T\}$. As for the core of $g_K$, we have the following result.
\begin{lemma}
The core of $g_K$ is given by, for any $K\in\binom{N}{k}$,
\[
C(g_K) = \left\{x\in\mathbb{R}^n\left| \begin{array}{lr} x_i=0,& i\not\in K\\
0\leqslant x_i\leqslant k-2,& i\in K\\
\sum_{i\in K}x_i = k-1&
\end{array}\right.
\right\}
\]
\end{lemma}
\begin{proof}
By definition of $g_K$, we have $g_K(S)=k-1$ if $S\supseteq K$, $g_K(S)\geqslant 1$ if $S\supseteq K\setminus\{i\}$ for some $i$, and $g_K(S)=0$ otherwise.
Take $x\in C(g_K)$.
For every $i\in N$, we have $g_K(\{i\})=0$, which implies $x\in\mathbb{R}^n_+$. Therefore, the set of inequalities $x(S)\geqslant g_K(S)$, $S\subseteq N$, can be reduced to $x_i\geqslant 0$ for all $i\in N$, and  $x(K)\geqslant k-1$, $x(K\setminus \{i\})\geqslant 1$ for all $i\in K$. Now, by $x(N)=g_K(N)=k-1$, we deduce $x(K)=k-1$ and $x(N\setminus K)=0$, implying $x_i=0$ for all $i\not\in K$. Finally, $x(K\setminus i)\geqslant 1$ and $x(K)=k-1$ imply $x_i\leqslant k-2$. Conversely, taking $x_i\leqslant k-2$ for all $i\in K$ and $x(K)=k-1 $ imply $x(K\setminus i)\geqslant 1$ for all $i\in K$.
 \end{proof}
Since the core decomposition has $O(n^k)$ terms, each of which is a polytope described by at most $2k$ inequalities, and the extension complexity of a Minkowski sum of polytopes is at most the sum of the extension complexities of the summands, we can conclude the following.
\begin{corollary}
    For every fixed \(k\ge3\), the core of a \(k\)-additive \((k-1)\)-monotone game has an explicitly constructible extended formulation of size \(O(n^k)\).
\end{corollary}

\section{Discussion and open directions}

Our results resolve the complexity of the three main regimes of bounded-additivity ($k$-additivity for fixed $k$)
cooperative games and introduce extended formulations as a tool in this
setting.  Several directions remain open.

\begin{enumerate}
\item \textbf{Extended formulations for cores at lower monotonicity
orders.}
For fixed $k\ge 4$ and $2\le l\le k-2$, do the cores of
$k$-additive $l$-monotone games always admit extended formulations
of size polynomial in $n$? Our exponential lower bounds concern
the cones of games and do not answer this question for the cores
of individual games. Linear optimization over these cores remains
polynomial-time solvable by supermodularity.

\item \textbf{Other solution concepts.}
      The facial-reduction strategy used here — realizing a hard polytope in
      the dual and propagating hardness — is general.  It would be interesting
      to apply it to other solution concepts in cooperative games (in particular, the nucleolus and the kernel) and to other bounded-degree
      classes of set functions such as capacities.

\item \textbf{Better Hardness proofs.}
      Our hardness proofs in Theorem \ref{thm:hardness} use ellipsoid algorithm implicitly and thus are weakly polynomial reductions. Can one give a strongly polynomial hardness reduction?
      
\item \textbf{Practical LP formulations.}
      The extended formulation for the core in result~(3) is explicit and of
      size $O(n^k)$.  Evaluating its practical performance for moderate $n$
      and $k$ — compared with direct separation-based approaches — is a
      natural experimental question.
\end{enumerate}

\section*{Acknowledgements}
Generative AI, specifically Perplexity and ChatGPT Astra, were used to help with grammar and linguistic style. The technical parts of the paper were drafted entirely by the authors except for Subsection \ref{subsec:core_k_k-1} which was pointed out by Astra when asked for opinion on the state of the paper. The authors take full responsibility for correctness of the results.
\bibliographystyle{plain}
\bibliography{refs}

\end{document}

%% file: config.tex
\usepackage{pgffor}
\foreach \x in {A,...,Z}{%
	\expandafter\xdef\csname b\x\endcsname{\noexpand\ensuremath{\noexpand\mathbf{\x}}}
	\expandafter\xdef\csname c\x\endcsname{\noexpand\ensuremath{\noexpand\mathcal{\x}}}
	\expandafter\xdef\csname B\x\endcsname{\noexpand\ensuremath{\noexpand\mathbb{\x}}}
}

\foreach \x in {a,...,z}{%
	\expandafter\xdef\csname b\x\endcsname{\noexpand\ensuremath{\noexpand\mathbf{\x}}}
}

\newcommand{\conv}[1]{\ensuremath{\mathrm{conv}\left(#1\right)}}
\newcommand{\cone}[1]{\ensuremath{\mathrm{cone}\left(#1\right)}}

\newcommand{\xc}[1]{\ensuremath{\mathrm{xc}\left(#1\right)}}
\newcommand{\core}{\mathrm{Core}}
\newcommand{\cor}{\mathrm{CORR}}

\newcommand{\Oh}[1]{\ensuremath{\cO}\left(#1\right)}
\newcommand{\polar}{\ensuremath{\Delta}}

\newtheorem{problem}{Problem}
\newtheorem*{problem*}{Problem}
\newtheorem{obs}{Observation}

\newcommand{\SM}[2]{\ensuremath{\cM}^{#1}_{#2}}